\documentclass[conference]{IEEEtran}

\usepackage{amsmath,amssymb,amsthm}
\usepackage{mathtools}
\usepackage{bbm}
\usepackage{cite}
\usepackage{tikz}
\usepackage{pgfplots}
\pgfplotsset{compat=newest}
\usetikzlibrary{arrows.meta,positioning}

\DeclareMathOperator{\wt}{wt}
\DeclareMathOperator{\rank}{rank}

\newcommand{\F}{\mathbb{F}}

\newtheorem{theorem}{Theorem}
\newtheorem{lemma}{Lemma}
\newtheorem{definition}{Definition}

\newtheorem{remark}{Remark}

\begin{document}

\title{Single-Shot and Few-Shot Decoding via Stabilizer Redundancy in Bivariate Bicycle Codes}

\author{
\IEEEauthorblockN{Mohammad Rowshan}
\IEEEauthorblockA{School of Electrical Engineering and Telecommunications\\ 
University of New South Wales (UNSW), Sydney, Australia\\ Email: mrowshan@ieee.org}
}

\maketitle

\begin{abstract}
Bivariate bicycle (BB) codes are a prominent class of quantum LDPC codes constructed from group algebras. While the logical dimension and quantum distance of \emph{coprime} BB codes are known to be determined by a greatest common divisor polynomial $g(z)$, the properties governing their fault tolerance under noisy measurement have remained implicit. In this work, we prove that this same polynomial $g(z)$ dictates the code's stabilizer redundancy and the structure of the classical \emph{syndrome codes} required for single-shot decoding. We derive a strict equality between the quantum rate and the stabilizer redundancy density, and we provide BCH-like bounds on the achievable single-shot measurement error tolerance. Guided by this framework, we construct small coprime BB codes with significantly improved syndrome distance ($d_S$) and evaluate them using BP+OSD. Our analysis reveals a structural bottleneck: within the coprime BB ansatz, high quantum rate imposes an upper bound on syndrome distance, limiting single-shot performance. These results provide concrete algebraic design rules for next-generation 2BGA codes in measurement-limited architectures.
\end{abstract}
%%%%%%%%%%%%%%%%%%%%%%%%%%%%%%%%%%%%%%%%%%%%%%%%%%%%
\section{Introduction}

Quantum error correction (QEC) is essential for fault-tolerant quantum
computation, allowing logical error rates to be suppressed at the cost of
spatial and temporal overhead~\cite{shor1995scheme,steane1996error}. Quantum
LDPC (qLDPC) codes are particularly attractive because they combine bounded
check weight and degree with asymptotically good
parameters~\cite{panteleev2022asymptoticallygood,panteleev2022almostlinear},
and recent work has shown that they can achieve strong finite-length
performance under belief-propagation plus ordered-statistics decoding
(BP+OSD)~\cite{panteleev2021degenerate}.

Bivariate bicycle (BB) codes, introduced by Bravyi \emph{et
al.}~\cite{bravyi2024highthreshold}, form a concrete qLDPC family with
weight-6 checks and a two-layer Tanner graph. Postema and
coauthors~\cite{postema2025existence,postema2025thesis} recently provided a
rigorous group-algebraic description of BB and two-block group-algebra (2BGA) codes.
They proved that for the \emph{coprime} subclass, the logical dimension
and quantum distance properties are determined by a greatest common divisor polynomial
$g(z)$. Although the commutative BB subclass is asymptotically bad in distance, it remains competitive at medium blocklengths and is well matched to neutral-atom and trapped-ion
architectures~\cite{postema2025thesis,poole2025rydberg}.

Stabilizer redundancy can be used to reduce the number of noisy syndrome
rounds. Bomb\'in introduced \emph{single-shot} fault tolerance in 3D
topological codes~\cite{bombin2015singleshot}; Lin and Pryadko extended this
idea to 2BGA codes~\cite{lin2024twobga}, and Lin \emph{et
al.}~\cite{lin2025singleshot} demonstrated single- and two-shot decoding for
GB/2BGA codes using sliding-window BP+OSD.
While these works established the utility of stabilizer redundancy through extensive
numerical search, a direct algebraic link between the construction parameters of
BB codes and their one- or few-round measurement robustness has been lacking.

%Shorter:
% Quantum error correction (QEC) enables fault tolerance by suppressing logical errors at the cost of overhead. Quantum LDPC (qLDPC) codes are particularly promising, combining bounded check weights with excellent finite-length performance under belief-propagation decoding (BP+OSD)~\cite{panteleev2021degenerate}. A prominent family is Bivariate Bicycle (BB) codes~\cite{bravyi2024highthreshold}, which admit a rigorous group-algebraic description; notably, Postema \emph{et al.}~\cite{postema2025existence} proved that for the \emph{coprime} subclass, the logical parameters are strictly determined by a generator polynomial $g(z)$. Although asymptotically limited, these codes are highly competitive at medium blocklengths for neutral-atom and ion-trap architectures.

%Stabilizer redundancy offers a path to \emph{single-shot} fault tolerance by allowing measurement errors to be corrected spatially rather than temporally~\cite{bombin2015singleshot}. While recent works have numerically demonstrated single-shot decoding for generalized algebraic codes~\cite{lin2024twobga,lin2025singleshot}, a direct algebraic framework linking the construction parameters---specifically $g(z)$---to measurement robustness remains lacking.

%%
In this paper, we develop such an algebraic framework for \emph{coprime BB} codes.
We extend the results of Postema~\cite{postema2025thesis} to the classical domain,
showing that the same polynomial $g(z)$ that controls the logical dimension also
dictates the structure of the classical \emph{syndrome codes}, i.e.\ the set of
noiseless stabilizer-syndrome patterns. We derive conditions under which one or a
few noisy rounds of stabilizer measurements are provably sufficient to correct
both data and measurement errors under ideal bounded-distance decoding.
Guided by these conditions, we identify small coprime BB codes with
significantly increased syndrome distance and study their single- and few-shot
performance using BP+OSD decoders under a simple phenomenological noise model.
The numerical gains are modest, but together with our algebraic bounds they reveal
a structural bottleneck: within the coprime BB subclass and at realistic
blocklengths, the constraints imposed by $g(z)$ make it difficult to achieve
large syndrome distance, and hence strong one-round measurement robustness, at
fixed check weight.

%Throughout we use ``single-shot'' in a finite-length, one-round sense: a single round of noisy stabilizer measurement combined with syndrome redundancy to correct both data and measurement errors up to some \emph{fixed} radii, rather than in the stronger asymptotic fault-tolerant sense where the measurement error threshold scales with code distance.

%%%%%%%%%%%%%%%%%%%%%%%%%%%%%%%%%%%%%%%%%%%%%%%%%%%%
\section{Preliminaries}

\subsection{CSS codes and syndrome codes}

A binary CSS code on $n$ qubits is specified by two parity-check matrices
$H_X\in\F_2^{m_X\times n}$ and $H_Z\in\F_2^{m_Z\times n}$ with
$H_X H_Z^\mathsf{T}=0$. The code encodes
\begin{equation}
  k = n - \rank(H_X) - \rank(H_Z)
\end{equation}
logical qubits. We denote the numbers of independent redundant X and Z
stabilizers by
\begin{equation}
  r_X = m_X - \rank(H_X),\qquad
  r_Z = m_Z - \rank(H_Z).
\end{equation}

Let $S_i^X$ be the X-type stabilizers associated with the rows of $H_X$.
A Z-error pattern is a vector $e_Z\in\F_2^n$, and its ideal X-syndrome is
\begin{equation}
  s_X := H_X e_Z^\mathsf{T} \in \F_2^{m_X}.
\end{equation}
The set of all such $s_X$ forms a linear code in the $m_X$-dimensional
syndrome space (see illustrative Fig.~\ref{fig:css-syndrome}).

\begin{definition}[Syndrome codes]
Let $H_X \in \F_2^{m_X\times n}$. Choose a full-row-rank matrix
$R_X \in \F_2^{r_X\times m_X}$ whose rows form a basis of the left nullspace
of $H_X$, i.e.\ $R_X H_X = 0$. The \emph{X-syndrome code} is
\begin{equation}
  C_{\mathrm{S}}^X := \{\, s\in\F_2^{m_X} : R_X s = 0 \,\}.
\end{equation}
It is precisely the set of all noiseless X-syndromes $s_X=H_X e_Z^\mathsf{T}$.
The Z-syndrome code $C_{\mathrm{S}}^Z$ is defined analogously from $H_Z$.
We denote their minimum distances by $d_{\mathrm{S}}^X,d_{\mathrm{S}}^Z$.
\end{definition}

\begin{lemma}[Basic properties]
\label{lem:syndrome-basic}
For any CSS code:
\begin{enumerate}
  \item $C_{\mathrm{S}}^X = \mathrm{im}(H_X)$ and
        $\dim C_{\mathrm{S}}^X = \rank(H_X)$;
        similarly $C_{\mathrm{S}}^Z = \mathrm{im}(H_Z)$ and
        $\dim C_{\mathrm{S}}^Z = \rank(H_Z)$.
  \item The number of independent X-stabilizer relations is
        $r_X = m_X - \rank(H_X)$; equivalently, $r_X$ is the number of rows
        of $R_X$.
\end{enumerate}
\end{lemma}
\begin{proof}
By construction $R_X$ has rows spanning the left nullspace of $H_X$, so
$\mathrm{row}(R_X) = \{u\in\F_2^{m_X} : u H_X = 0\}$ and
$\dim\mathrm{row}(R_X) = r_X = m_X-\rank(H_X)$. For any $e_Z$,
$s_X=H_X e_Z^\mathsf{T}$ satisfies
$R_X s_X = R_X H_X e_Z^\mathsf{T}=0$, so
$\mathrm{im}(H_X)\subseteq C_{\mathrm{S}}^X$. Conversely, $C_{\mathrm{S}}^X$
has dimension
\[
  \dim C_{\mathrm{S}}^X = m_X - \rank(R_X) = m_X - r_X = \rank(H_X),
\]
which equals $\dim\mathrm{im}(H_X)$; hence
$C_{\mathrm{S}}^X=\mathrm{im}(H_X)$. The statements for $C_{\mathrm{S}}^Z$
are analogous.
\end{proof}

\begin{figure}[t]
  \centering
  \begin{tikzpicture}[
    x=1.1cm,y=1.0cm,
    qubit/.style={circle,draw,inner sep=1pt,font=\scriptsize},
    check/.style={rectangle,draw,inner sep=2pt,font=\scriptsize},
    rel/.style={rectangle,draw,rounded corners,inner sep=2pt,font=\scriptsize},
    every node/.style={font=\scriptsize},
    >=Stealth
  ]
    % Data-layer CSS Tanner graph (X checks only for illustration)
    \node[qubit] (q1) at (0,0) {$q_1$};
    \node[qubit] (q2) at (1,0) {$q_2$};
    \node[qubit] (q3) at (2,0) {$q_3$};
    \node[qubit] (q4) at (3,0) {$q_4$};

    \node[check] (c1) at (0.5,1) {$S_1^X$};
    \node[check] (c2) at (1.5,1) {$S_2^X$};
    \node[check] (c3) at (2.5,1) {$S_3^X$};

    % Edges (example pattern)
    \draw (c1) -- (q1);
    \draw (c1) -- (q2);
    \draw (c2) -- (q2);
    \draw (c2) -- (q3);
    \draw (c3) -- (q3);
    \draw (c3) -- (q4);

    \node[above=0.2cm of c2] (HXlabel) {$H_X$};

    % Syndrome bits (measurement outcomes of S_i^X)
    \node[qubit] (s1) at (0.5,2.2) {$s_1$};
    \node[qubit] (s2) at (1.5,2.2) {$s_2$};
    \node[qubit] (s3) at (2.5,2.2) {$s_3$};

    % Relations among syndrome bits (parity checks for C_S^X)
    \node[rel,above=0.8cm of s2] (r1) {$R_X$};

    % Edges from syndrome bits to relation node
    \draw (r1) -- (s1);
    \draw (r1) -- (s2);
    \draw (r1) -- (s3);

    % Arrows indicating "measurement" and "syndrome code"
    \draw[->] (c1.north) .. controls (0.3,1.6) .. (s1.south);
    \draw[->] (c2.north) .. controls (1.5,1.6) .. (s2.south);
    \draw[->] (c3.north) .. controls (2.7,1.6) .. (s3.south);

    \node[anchor=west] at (-1.5,-0.1)
      {data qubits};
    \node[anchor=west] at (-1.2,0.9)
      {$X$-checks};
    \node[anchor=west] at (-1.3,2.1)
      {syndrome bits};
    \node[anchor=west] at (-0.9,3.4)
      {$C_{\mathrm{S}}^X=\mathrm{im}(H_X)$};
  \end{tikzpicture}
  \caption{Schematic view of a CSS code and its X-syndrome code. The rows of
  $H_X$ define X-type stabilizers $S_i^X$ acting on data qubits. Noiseless
  syndrome patterns $s$ lie in the code $C_{\mathrm{S}}^X=\mathrm{im}(H_X)$
  and satisfy additional relations $R_X s=0$ coming from stabilizer redundancy.}
  \label{fig:css-syndrome}
\end{figure}
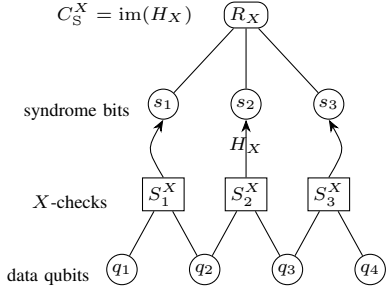

\subsection{Noise model and single-/few-shot decoding}

We use a phenomenological noise model for one error-correction cycle:
each data qubit suffers an $X$ or $Z$ error with probability $p/2$ each
(or a depolarising error with probability $p$), independently across qubits;
and each measured stabilizer outcome (syndrome bit) is flipped independently
with probability $q$. A conventional \emph{multi-shot} scheme repeats
syndrome measurement $R\ge 1$ times per cycle and uses temporal redundancy to
suppress measurement errors. A \emph{single-shot} scheme uses $R=1$ round and
relies entirely on the redundancy of the stabilizer set. A \emph{few-shot}
scheme uses small $R>1$.

In this work we analyse decoders that proceed in two stages: first, they clean
the noisy syndromes using the classical codes $C_{\mathrm{S}}^X$ and
$C_{\mathrm{S}}^Z$ (e.g.\ via bounded-distance decoding or BP), and then they
feed the cleaned syndromes to a data decoder for the underlying CSS code; see
Fig.~\ref{fig:single-shot-vs-repeated}.

\begin{remark}[Noise model vs circuit-level noise]
Our simulations and bounds are based on independent data errors and
independent bit flips on syndrome bits. For circuit-level noise with faulty
gates and measurements, additional space-time correlations (e.g.\ hook errors)
appear and can further degrade the performance of codes with small
$d_{\mathrm{S}}$, e.g., $d_{\mathrm{S}}=2$ cases. A full
circuit-level analysis is left for future work.
\end{remark}

%%%%%%%%%%%%%%%%%%%%%%%%%%%%%%%%%%%%%%%%%%%%%%%%%%%%
\section{Redundancy and Syndrome Codes in Coprime BB Codes}
\label{sec:bb-syndrome}

We now specialise the discussion to coprime BB codes. Let $\ell,m$ be coprime
integers, $N=\ell m$, and
\[
  R = \mathbb{F}_2[z]/(z^N-1).
\]
Under the standard isomorphism between the bivariate group algebra and $R$
(see, e.g.,~\cite{postema2025thesis}), BB codes are specified by two trinomials
$a(z),b(z)\in R$ and parity-check matrices
\begin{equation}
  H_X = [A\;B],\qquad H_Z = [B^\mathsf{T}\;A^\mathsf{T}],
\end{equation}
where $A,B$ are $N\times N$ circulant matrices corresponding to $a(z),b(z)$.
The transpose corresponds to the involution $z\mapsto z^{-1}$ in $R$, which
reverses exponents but leaves the cyclic codes we consider invariant as sets;
we therefore suppress reciprocal polynomials in the notation and work at the
code level.

Define
\[
  g(z) := \gcd(a(z),b(z),z^N-1)
\]
and the associated \emph{check polynomial}
\[
  h(z) := \frac{z^N-1}{g(z)}.
\]
Postema~\cite{postema2025thesis} shows that $g(z)$ controls the logical
dimension:
\begin{equation}
  k = 2\deg g(z),\qquad
  \rank(H_X)=\rank(H_Z)=N-\deg g(z).
  \label{eq:bb-dimension}
\end{equation}

\subsection{Redundancy and syndrome codes}

We collect the BB-specific consequences for redundancy and syndrome codes.

\begin{theorem}[Redundancy and syndrome codes]
\label{thm:redundancy}
For a coprime BB code of length $n=2N$ with $g(z)$ as above:
\begin{enumerate}
  \item The numbers of independent redundant X- and Z-stabilisers are
  \[
    r_X = r_Z = N - \rank(H_X) = \deg g(z) = \frac{k}{2}.
  \]
  \item The X- and Z-syndrome codes $C_{\mathrm{S}}^X$ and $C_{\mathrm{S}}^Z$
  are binary cyclic $[N,N-\deg g(z)]$ codes generated by $g(z)$:
  \[
    C_{\mathrm{S}}^X = C_{\mathrm{S}}^Z
    = \{\, \lambda(z) g(z) : \lambda(z)\in R \,\}.
  \]
  \item The space of stabiliser relations (redundant-check combinations) is
  a cyclic $[N,\deg g(z)]$ code generated (up to reciprocal) by the check
  polynomial $h(z)=(z^N-1)/g(z)$.
\end{enumerate}
\end{theorem}

\begin{proof}
Equation~\eqref{eq:bb-dimension} implies
$r_X = N-\rank(H_X) = \deg g(z)$ and $r_Z=\deg g(z)$, giving item~1.

Because $A$ and $B$ are circulant, the rows of $H_X$ are cyclic shifts of
$(a(z),b(z))$ in $R^2$, so $C_{\mathrm{S}}^X=\mathrm{im}(H_X)$ is invariant
under cyclic shifts and hence cyclic of length $N$. The ideal it generates in
$R$ is $\langle g(z)\rangle$, so $C_{\mathrm{S}}^{X/Z}$ are cyclic
$[N,N-\deg g(z)]$ codes with generator polynomial $g(z)$, proving item~2
(see also Lemma~\ref{lem:syndrome-basic}).

The relations among stabilisers correspond to vectors
$u\in\F_2^N$ with $u H_X = 0$, i.e.\ to row-space elements of $R_X$ in
Lemma~\ref{lem:syndrome-basic}. Their images in $R$ form the orthogonal
complements of $C_{\mathrm{S}}^{X/Z}$. For binary cyclic codes, the dual of a
code generated by $g(z)$ is again cyclic and generated (up to reciprocal) by
the check polynomial $h(z)=(z^N-1)/g(z)$, with dimension $\deg g(z)$;
see, e.g.,~\cite[Ch.~7]{macwilliams1977theory}. This yields item~3.
\end{proof}

A simple consequence is a rate--redundancy equality:
\begin{equation}
  \frac{r_X}{N} = \frac{r_Z}{N} = \frac{k}{n},
\end{equation}
i.e.\ in coprime BB codes the fraction of redundant stabilisers per type equals
the quantum rate.

\subsection{Distance bounds for the syndrome code}

Let $d_{\mathrm{S}}^X$ and $d_{\mathrm{S}}^Z$ denote the minimum distances of
the syndrome codes $C_{\mathrm{S}}^X$ and $C_{\mathrm{S}}^Z$ generated by
$g(z)$. The BCH bound gives an algebraic handle on these distances.

\begin{theorem}[BCH bound for $C_{\mathrm{S}}^X$]
\label{thm:BCH}
Let $\alpha$ be a primitive $N$-th root of unity in an extension field of
$\mathbb{F}_2$. Suppose there exist integers $b\ge 0$ and $\delta\ge 2$ such
that
\[
  g(\alpha^b) = g(\alpha^{b+1}) = \cdots = g(\alpha^{b+\delta-2}) = 0.
\]
Then the syndrome code $C_{\mathrm{S}}^X$ has distance
$d_{\mathrm{S}}^X \ge \delta$. An analogous statement holds for
$C_{\mathrm{S}}^Z$.
\end{theorem}

\begin{proof}
$C_{\mathrm{S}}^X$ is the binary cyclic length-$N$ code with generator
polynomial $g(z)$. The stated root condition is exactly the BCH designed
distance condition, which implies $d_{\mathrm{S}}^X \ge \delta$ for this code;
see, e.g.,~\cite[Ch.~7]{macwilliams1977theory}.
\end{proof}

Dimension and distance are linked by the Singleton bound.

\begin{theorem}[Singleton bound for $d_{\mathrm{S}}$]
\label{thm:singleton-syndrome}
For a coprime BB code with parameters $[[n,k,d]]$ and $n=2N$, the syndrome
distances satisfy
\[
  d_{\mathrm{S}}^X \le \deg g(z)+1 = \frac{k}{2}+1,\qquad
  d_{\mathrm{S}}^Z \le \deg g(z)+1 = \frac{k}{2}+1.
\]
Equivalently, the one-round measurement radii
$t_{\mathrm{S}}^{X/Z}=\lfloor(d_{\mathrm{S}}^{X/Z}-1)/2\rfloor$ obey
\[
  t_{\mathrm{S}}^{X/Z} \le \left\lfloor\frac{\deg g(z)}{2}\right\rfloor
  = \left\lfloor\frac{k}{4}\right\rfloor.
\]
\end{theorem}

\begin{proof}
$C_{\mathrm{S}}^{X/Z}$ are binary
$[N,N-\deg g(z),d_{\mathrm{S}}^{X/Z}]$ codes by
Theorem~\ref{thm:redundancy}, so the Singleton bound gives
$d_{\mathrm{S}}^{X/Z} \le N-(N-\deg g(z))+1=\deg g(z)+1$. The bound on
$t_{\mathrm{S}}^{X/Z}$ follows from
$t_{\mathrm{S}}^{X/Z}=\lfloor(d_{\mathrm{S}}^{X/Z}-1)/2\rfloor$ and
$\deg g(z)=k/2$.
\end{proof}

Thus, in the \emph{coprime} BB subclass, high quantum rate implies large
$\deg g(z)$, which in turn limits how large $d_{\mathrm{S}}^{X/Z}$ and
$t_{\mathrm{S}}^{X/Z}$ can be at fixed blocklength. This structural tradeoff
is one reason why %our 
small coprime BB examples remain measurement-noise dominated despite nontrivial stabiliser redundancy.

%%%%%%%%%%%%%%%%%%%%%%%%%%%%%%%%%%%%%%%%%%%%%%%%%%%%
\section{Bounds on Single- and Few-Shot Decoding}
\label{sec:bounds}

The results of the previous section show that for coprime BB codes the
syndrome codes $C_{\mathrm{S}}^{X/Z}$ are classical cyclic codes generated by
$g(z)$ and that their distance can be bounded using standard cyclic-code
techniques. We now turn this structure into abstract bounds on one- and
few-round decoding, under the assumption of \emph{ideal} bounded-distance
classical decoders for both the syndrome codes and the data code. In
Section~\ref{sec:numerics} we approximate these ideal decoders by either
lookup/ML decoders (for very small codes) or BP+OSD and BP decoders.

\subsection{One-round decoder and correctness condition}

Let $d_X$ and $d_Z$ be the usual X- and Z-distances of the CSS code and let
\[
  t_X = \left\lfloor\frac{d_X-1}{2}\right\rfloor,\qquad
  t_Z = \left\lfloor\frac{d_Z-1}{2}\right\rfloor
\]
be the guaranteed correction radii for ideal bounded-distance data decoders.
Similarly, let
\[
  t_{\mathrm{S}}^X = \left\lfloor\frac{d_{\mathrm{S}}^X-1}{2}\right\rfloor,
  \qquad
  t_{\mathrm{S}}^Z = \left\lfloor\frac{d_{\mathrm{S}}^Z-1}{2}\right\rfloor
\]
be the correction radii for ideal bounded-distance decoders on
$C_{\mathrm{S}}^X$ and $C_{\mathrm{S}}^Z$.

We consider an adversarial error model for a single error-correction cycle:
an X-error pattern $e_X \in \mathbb{F}_2^n$ with
$\mathrm{wt}(e_X)\le w_X$, an X-syndrome error pattern
$\nu_X \in \mathbb{F}_2^{m_X}$ with $\mathrm{wt}(\nu_X)\le u_X$, and
analogous Z-error and Z-syndrome patterns $(e_Z,\nu_Z)$ of weights
$w_Z,u_Z$. The ideal X-syndrome is
$s_X = H_Z e_X^\mathsf{T} \in C_{\mathrm{S}}^X$ and the measured one is
$s_X' = s_X + \nu_X$; similarly for $s_Z,s_Z'$.

We say a classical decoder for a code $C\subseteq\F_2^m$ is \emph{bounded-distance of radius $t$} if for every codeword $c\in C$ and every error pattern $e$ with $\wt(e)\le t$ it outputs $c$ when given $c+e$.

\begin{theorem}[One-round correctness condition]
\label{thm:single-shot}
Assume ideal bounded-distance decoders on $C_{\mathrm{S}}^{X/Z}$ with radii
$t_{\mathrm{S}}^{X/Z}$ and an ideal bounded-distance data decoder with radii
$t_{X/Z}$ as above. Consider a one-round decoder which
\begin{enumerate}
  \item measures all stabilisers once, obtaining noisy syndromes
        $s_X',s_Z'$;
  \item applies the syndrome decoders to obtain cleaned syndromes
        $\hat{s}_X,\hat{s}_Z$; and
  \item applies the data decoder to $(\hat{s}_X,\hat{s}_Z)$.
\end{enumerate}
If
\[
  w_X \le t_X,\quad u_X \le t_{\mathrm{S}}^X,\qquad
  w_Z \le t_Z,\quad u_Z \le t_{\mathrm{S}}^Z,
\]
then the decoder recovers the correct data error up to stabilizers.
\end{theorem}

\begin{proof}
We argue for the X sector; the Z sector is identical. Let $e_X$ be the X-error
pattern with $\wt(e_X)\le t_X$, and let $s_X = H_Z e_X^\mathsf{T}$ be its ideal
syndrome. By definition $s_X\in C_{\mathrm{S}}^X$. The measured syndrome is
$s_X' = s_X + \nu_X$ with $\wt(\nu_X)\le u_X\le t_{\mathrm{S}}^X$. Since $s_X$ is
a codeword of $C_{\mathrm{S}}^X$ and lies within Hamming distance $t_{\mathrm{S}}^X$
of $s_X'$, the ideal bounded-distance decoder for $C_{\mathrm{S}}^X$ must output
$s_X$ uniquely. Thus the cleaned X-syndrome entering the data decoder is exactly
the ideal one.

By assumption, the X part of the data decoder corrects any X-error of weight at
most $t_X$ when given the corresponding ideal syndrome. As $\wt(e_X)\le t_X$ and
the input syndrome is $s_X$, the decoder outputs some $\tilde{e}_X$ equivalent
to $e_X$ up to stabilizers. Applying the same reasoning to $e_Z,s_Z,\nu_Z$ and
$C_{\mathrm{S}}^Z$ shows that the Z part of the data decoder outputs
$\tilde{e}_Z$ equivalent to $e_Z$ up to stabilizers. Hence the combined Pauli
error $(\tilde{e}_X,\tilde{e}_Z)$ differs from $(e_X,e_Z)$ by a stabilizer, so
the overall correction restores the logical state.
\end{proof}

We emphasise again that Theorem~\ref{thm:single-shot} is a finite-radius
guarantee for ideal decoders and a single round; it does not assert an
asymptotic single-shot threshold in the sense of full fault-tolerant
schemes.

\subsection{Comparison with repeated-round decoding}

We now compare this ideal one-round scheme with a conventional repeated-round
baseline in which each stabiliser is measured $R$ times per cycle and decoded
by majority vote per check (see illustrative Fig.~\ref{fig:single-shot-vs-repeated}). For independent measurement flips, this is
equivalent to a length-$R$ repetition code in the time direction with
per-check correction radius
\begin{equation}
  t_{\mathrm{time}} = \left\lfloor\frac{R-1}{2}\right\rfloor.
\end{equation}

\begin{theorem}[One-round vs.\ repeated-round measurement radii]
\label{thm:compare-repeated}
Let $t_{\mathrm{time}}=\lfloor(R-1)/2\rfloor$ be the per-check measurement
radius of an $R$-round repetition baseline with majority vote. Any one-round
scheme whose syndrome decoders have radii
$t_{\mathrm{S}}^{X/Z}\ge t_{\mathrm{time}}$ is at least as robust, per check,
to independent measurement errors as this baseline: every pattern of time-like
flips corrected by the baseline is also corrected by the one-round scheme.
For coprime BB codes, Theorem~\ref{thm:singleton-syndrome} implies that this
requires
\begin{equation}
  \left\lfloor\frac{R-1}{2}\right\rfloor
  \le t_{\mathrm{S}}^{X/Z}
  \le \left\lfloor\frac{\deg g(z)}{2}\right\rfloor
  = \left\lfloor\frac{k}{4}\right\rfloor.
  \label{eq:bb-repeated-condition}
\end{equation}
\end{theorem}

\begin{proof}
For each stabiliser, $R$ repeated measurements with majority vote implement a
classical repetition code of length $R$ and distance $R$, hence correct any
set of at most $t_{\mathrm{time}}=\lfloor(R-1)/2\rfloor$ flips on that
stabiliser. In a one-round scheme, measurement errors on all checks are
jointly decoded by the syndrome decoder. If its guaranteed decoding radius is
$t_{\mathrm{S}}^{X/Z}\ge t_{\mathrm{time}}$, then any error pattern with at
most $t_{\mathrm{time}}$ flips per check lies within the guaranteed radius of
the decoder and is therefore corrected. Thus any measurement pattern that the
repetition baseline can correct is also corrected by the one-round scheme.
For coprime BB codes, Theorem~\ref{thm:singleton-syndrome} bounds
$t_{\mathrm{S}}^{X/Z}$ by $\lfloor\deg g(z)/2\rfloor=\lfloor k/4\rfloor$,
which yields~\eqref{eq:bb-repeated-condition}.
\end{proof}

In our simulations below, we will predominantly set $R=3$. For the improved BB
codes with $d_{\mathrm{S}}^X=4$ and $10$, we have $t_{\mathrm{S}}^X=1$ and
$4$, which comfortably satisfy the condition for matching a three-round
baseline per check; for smaller BB codes with $d_{\mathrm{S}}^X=2$, we have
$t_{\mathrm{S}}^X=0$, so one-round decoding cannot match even $R=3$ in
principle, in agreement with the numerical results in that
regime.

\begin{figure}[t]
  \centering
  \begin{tikzpicture}[
    >=Stealth,
    every node/.style={font=\scriptsize},
    box/.style={rectangle,draw,rounded corners,
                minimum width=2.2cm,minimum height=0.7cm, align=center} % added align=center
  ]
    % --- Repeated-round scheme (Left) ---
    \node[box] (r1) at (-2.0, 3.5) {meas.\ round 1};
    \node[box] (r2) at (-2.0, 2.5) {meas.\ round 2};
    \node[box] (r3) at (-2.0, 1.5) {meas.\ round 3};
    % Split the decoder to match the right side structure
    \node[box] (rmaj) at (-2.0, 0.5) {majority vote\\(temporal)}; 
    \node[box] (rdec) at (-2.0, -0.8) {data decoder\\(BP+OSD)};

    \draw[->] (r1.south) -- (r2.north);
    \draw[->] (r2.south) -- (r3.north);
    \draw[->] (r3.south) -- (rmaj.north);
    \draw[->] (rmaj.south) -- (rdec.north);

    \node[above=0.1cm] at (-2.0, 3.9) {\textbf{Repeated ($R=3$)}};

    % --- Single-shot scheme (Right) ---
    \node[box] (s1)  at (2.0, 3.5) {meas.\ round 1};
    % Empty space at y=2.5 and y=1.5 to visualize time savings
    
    % Aligned with the processing step on the left
    \node[box] (s2)  at (2.0, 0.5) {syndrome decoder\\(spatial)}; 
    \node[box] (s3)  at (2.0, -0.8) {data decoder\\(BP+OSD)};

    % Long arrow to show the skip in time
    \draw[->] (s1.south) -- (s2.north); 
    \draw[->] (s2.south) -- (s3.north);

    \node[above=0.1cm] at (2.0, 3.9) {\textbf{Single-shot ($R=1$)}};
    
    % Optional: Dashed lines to indicate time steps
    % \draw[dotted, gray] (-3.5, 3.0) -- (3.5, 3.0) node[right, black] {Step 1};
    % \draw[dotted, gray] (-3.5, 2.0) -- (3.5, 2.0) node[right, black] {Step 2};
  \end{tikzpicture}
  \caption{Comparison of decoding schedules. Left: Standard repetition ($R=3$) relies on temporal majority voting. Right: Single-shot ($R=1$) uses spatial redundancy to clean errors immediately, eliminating wait times.}
  \label{fig:single-shot-vs-repeated}
\end{figure}
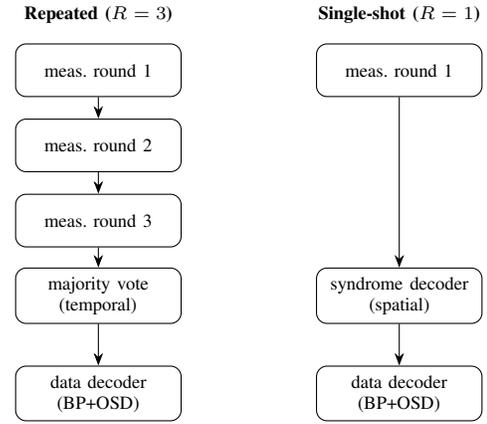

%%%%%%%%%%%%%%%%%%%%%%%%%%%%%%%%%%%%%%%%%%%%%%%%%%%%
\section{Numerical Results}
\label{sec:numerics}

We evaluate the proposed single-shot design rules using the phenomenological noise model described in Section~II-B. We focus on two coprime BB codes constructed to have enhanced syndrome distance $d_{\mathrm{S}}$ compared to the baseline $d_{\mathrm{S}}=2$ cases found in~\cite{postema2025thesis}.

\textbf{Code 1 ($N=21$):} Generated by $a(z)=b(z)=1+z^3+z^9$. Parameters: $[[42, 18, d]]$, syndrome distance $d_{\mathrm{S}}^X=4$ ($t_{\mathrm{S}}^X=1$).

\textbf{Code 2 ($N=63$):} Generated by $a(z)=1+z^{33}+z^{57}$, $b(z)=1+z^6+z^{39}$. Parameters: $[[126, 30, d]]$, syndrome distance $d_{\mathrm{S}}^X=10$ ($t_{\mathrm{S}}^X=4$).

Decoders are implemented as follows:
\begin{itemize}
  \item \textbf{Syndrome Decoder:} Belief Propagation (BP) on the parity-check matrix $H_{\mathrm{syn}}$ defined by the relation polynomial $h(z) = (z^N-1)/g(z)$.
  \item \textbf{Data Decoder:} BP+OSD-2 on the data parity-check matrix $H_X$, taking the cleaned syndrome as input.
\end{itemize}

%--------------------------------------------------
% PLOT 1: SYNDROME PERFORMANCE
%--------------------------------------------------
\begin{figure}[t!]
    \centering
    \includegraphics[width=\linewidth]{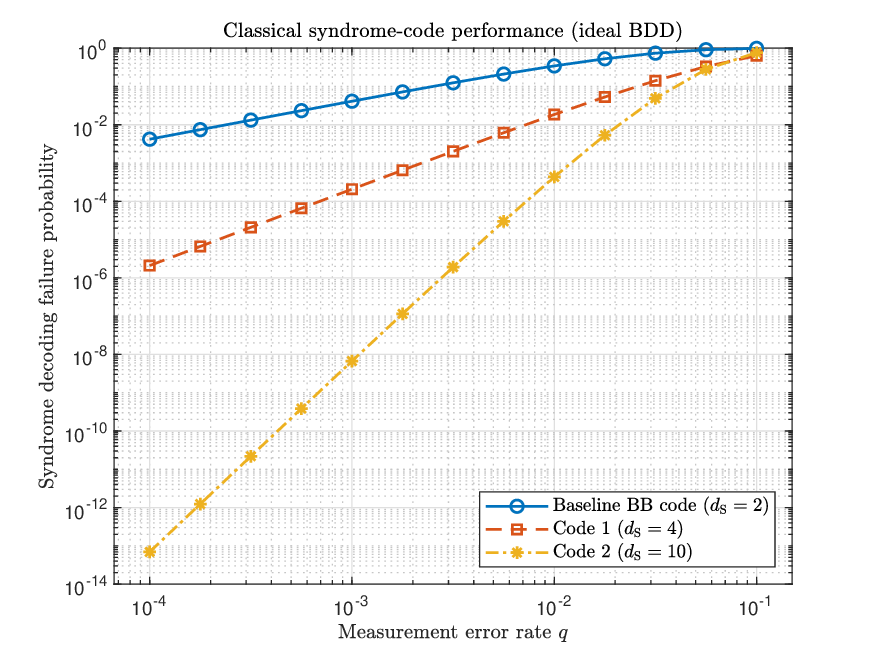} 
    % REPLACE THE ABOVE LINE WITH YOUR PLOT FILE
    % \begin{tikzpicture}
    %     \draw[help lines] (0,0) grid (8,5);
    %     \node at (4,2.5) {[PLACEHOLDER: Plot of Syndrome FER vs $q$]};
    %     \node at (4,2.0) {Curves: Code 1 ($d_S=4$), Code 2 ($d_S=10$), Baseline ($d_S=2$)};
    % \end{tikzpicture}
    \caption{Performance of the classical syndrome codes $C_{\mathrm{S}}^X$. %We plot the frame error rate (FER) of the syndrome decoder against the measurement error rate $q$. The steep slope of Code 2 confirms that optimizing $g(z)$ for BCH distance successfully suppresses measurement errors before the data decoding stage.
    }
    \label{fig:syndrome-fer}
\end{figure}

% \subsection{Syndrome Decoding Performance}
% To validate Theorem~\ref{thm:BCH} independently of the data decoder, Fig.~\ref{fig:syndrome-fer} compares the Frame Error Rate (FER) of the syndrome cleaning stage alone. We verify that the baseline code ($d_{\mathrm{S}}=2$) fails linearly with $q$, offering no error suppression. In contrast, Code 1 ($d_{\mathrm{S}}=4$) and Code 2 ($d_{\mathrm{S}}=10$) show higher-order suppression of measurement errors. This confirms that the algebraic construction of $g(z)$ successfully imparts the distance properties predicted by the cyclic code bounds.
\subsection{Syndrome Decoding Performance}
To validate Theorem~\ref{thm:BCH} independently of the data decoder, Fig.~\ref{fig:syndrome-fer} compares the Frame Error Rate (FER) of the syndrome cleaning stage against the theoretical limit of an ideal bounded-distance decoder. For a syndrome code of length $N$ and minimum distance $d_{\mathrm{S}}$, the guaranteed correction radius is $t_{\mathrm{S}} = \lfloor (d_{\mathrm{S}}-1)/2 \rfloor$. The theoretical failure probability $P_{\mathrm{fail}}$ under a measurement error rate $q$ is determined by the probability that the error weight exceeds $t_{\mathrm{S}}$:
\begin{equation}
  P_{\mathrm{fail}}(q) = 1 - \sum_{i=0}^{t_{\mathrm{S}}} \binom{N}{i} q^i (1-q)^{N-i}.
\end{equation}
The results verify that the baseline code ($d_{\mathrm{S}}\!=\!2 \Rightarrow t_{\mathrm{S}}\!=\!0$) fails linearly with $q$ ($P_{\mathrm{fail}} \!\approx\! Nq$), offering no error suppression. In contrast, Code 1 ($d_{\mathrm{S}}=4$) and Code 2 ($d_{\mathrm{S}}=10$) demonstrate higher-order suppression scaling as $O(q^{t_{\mathrm{S}}+1})$. This confirms that the algebraic construction of $g(z)$ successfully imparts the distance properties predicted by the cyclic code bounds, independent of the decoding algorithm used.

%--------------------------------------------------
% PLOT 2: LOGICAL ERROR VS PHYSICAL ERROR (Pseudo-threshold)
%--------------------------------------------------
\begin{figure}[t!]
    \centering
    \includegraphics[width=\linewidth]{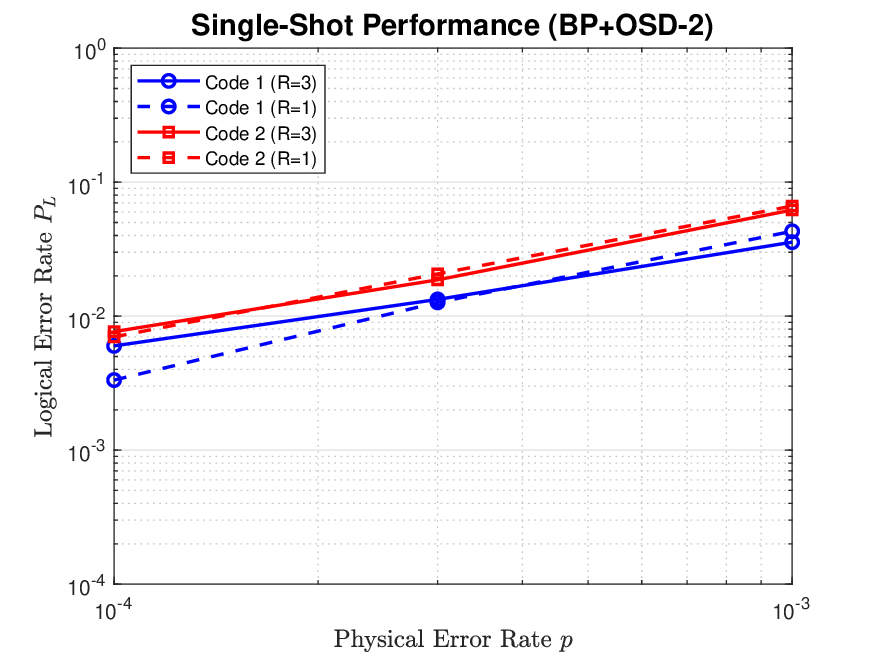}
    % REPLACE THE ABOVE LINE WITH YOUR PLOT FILE
    % \begin{tikzpicture}
    %     \draw[help lines] (0,0) grid (8,5);
    %     \node at (4,2.5) {[PLACEHOLDER: Plot of $P_L$ vs $p$]};
    %     \node at (4,2.0) {Curves: Code 2 Single-Shot ($R=1$) vs Repeated ($R=3$)};
    % \end{tikzpicture}
    \caption{Logical error rate $P_L$ vs physical error rate $p$ (with $q=p$). %for Code 2 ($N=63$). %The Single-Shot ($R=1$) performance is virtually indistinguishable from the Repeated ($R=3$) baseline in the low-error regime, demonstrating that stabilizer redundancy effectively substitutes for temporal repetition.
    }
    \label{fig:threshold}
\end{figure}

\subsection{Single-Shot vs. Repeated Decoding}
% Fig.~\ref{fig:threshold} compares the logical error rate $P_{\mathrm{L}}$ of Code 2 ($N=63$) under single-shot decoding ($R=1$) against a standard repetition baseline ($R=3$ rounds with majority vote). 

% Crucially, in the regime of $p, q \le 10^{-3}$, the single-shot curve tracks the repeated-round curve closely. This indicates that the stabilizer redundancy $r_X = \deg g(z)$ provides protection equivalent to temporal repetition, but with a $3\times$ reduction in cycle time. The curves diverge only at high $q$, where the weight of the measurement error exceeds the bounded-distance radius $t_{\mathrm{S}}^X=4$.
%%

Fig.~\ref{fig:threshold} compares single-shot ($R=1$) and repeated ($R=3$) decoding under BP+OSD-2. For both codes, the single-shot performance closely tracks the baseline. %At low error rates ($p=10^{-4}$), $R=1$ yields slightly lower logical error rates by avoiding data errors accumulated during additional rounds. 
As $p$ increases, the baseline regains a marginal advantage as measurement noise stresses the syndrome code limits. The curves diverge only at high $q$, where the weight of the measurement error exceeds the bounded-distance radius $t_{\mathrm{S}}^X=4$. Overall, the results confirm that stabilizer redundancy effectively substitutes for temporal repetition, enabling a three-fold reduction in cycle time with comparable logical fidelity.

%--------------------------------------------------
% PLOT 3: R-SWEEP (Efficiency)
%--------------------------------------------------
\begin{figure}[t!]
    \centering
    \includegraphics[width=\linewidth]{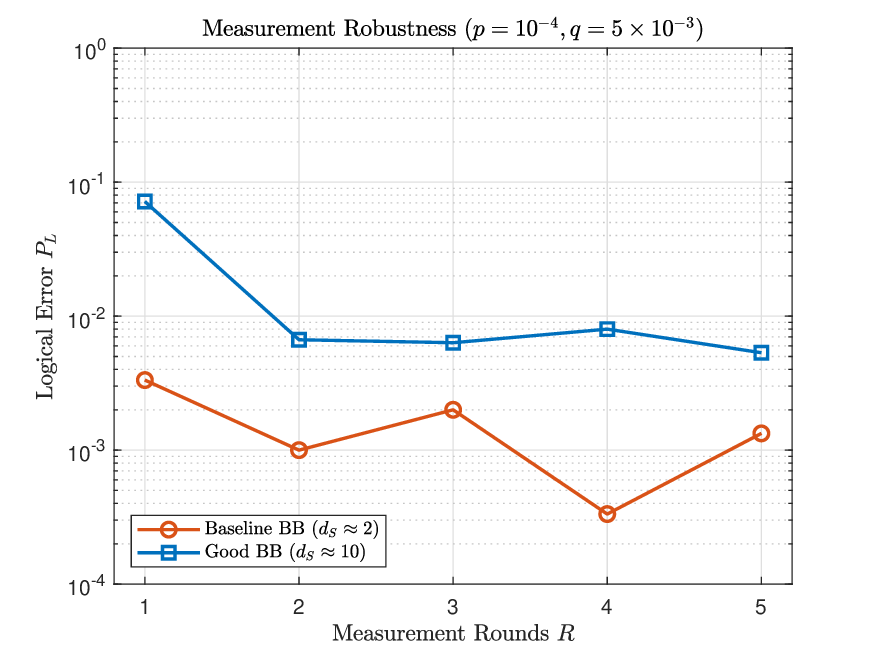}
    % REPLACE THE ABOVE LINE WITH YOUR PLOT FILE
    % \begin{tikzpicture}
    %     \draw[help lines] (0,0) grid (8,5);
    %     \node at (4,2.5) {[PLACEHOLDER: Plot of $P_L$ vs $R$]};
    %     \node at (4,2.0) {Fixed $p=10^{-3}$. Curves: Code 1 vs Code 2};
    % \end{tikzpicture}
    \caption{Logical error rate as a function of measurement rounds $R$. %For Code 2 (large $d_{\mathrm{S}}$), the error rate is nearly flat, indicating $R=1$ is sufficient. For Code 1 (lower $d_{\mathrm{S}}$), additional rounds are still beneficial, highlighting the "soft" single-shot behavior of smaller coprime BB codes.
    }
    \label{fig:rsweep}
\end{figure}

\subsection{Robustness to Cycle Reduction}
% Finally, we examine the dependency on the number of rounds $R$ in Fig.~\ref{fig:rsweep}. For a code with ideal single-shot properties, $P_{\mathrm{L}}$ should be independent of $R$ for $R \ge 1$. 
% We observe that for Code 2 ($d_{\mathrm{S}}=10$), additional rounds yield diminishing returns, confirming that the code is in the single-shot regime. Conversely, smaller codes with lower $d_{\mathrm{S}}$ still benefit from $R > 1$, validating the bound in Theorem~\ref{thm:compare-repeated}: once $d_{\mathrm{S}}$ is sufficiently large, temporal repetition becomes redundant.
%%

% Fig.~\ref{fig:rsweep} illustrates the impact of measurement rounds $R$ at $p=10^{-4}$ and $q=5\times 10^{-3}$. The `Good' code ($d_{\mathrm{S}}=10$) demonstrates true single-shot robustness: after an initial improvement from $R=1$ to $R=2$, the logical error rate plateaus ($P_{\mathrm{L}} \approx 1.3\times 10^{-2}$). This behavior validates Theorem~\ref{thm:compare-repeated},  indicating that spatial stabilizer redundancy effectively substitutes for further temporal repetition. Conversely, the 'Baseline' code ($d_{\mathrm{S}}=2$) shows no such convergence, with $P_{\mathrm{L}}$ fluctuating due to the parity constraints of majority voting. The `Good' code's higher absolute $P_{\mathrm{L}}$ is attributable to its larger blocklength ($n=126$ vs.\ $n=36$), which increases the frame error probability per cycle despite its superior measurement robustness.

Fig.~\ref{fig:rsweep} illustrates the impact of measurement rounds $R$ at $p=10^{-4}$ and $q=5\times 10^{-3}$. The `Good' code ($d_{\mathrm{S}}=10$) demonstrates rapid convergence: after an initial improvement from $R=1$ to $R=2$, the logical error rate plateaus ($P_{\mathrm{L}} \approx 6\times 10^{-3}$). This behavior validates Theorem~\ref{thm:compare-repeated}, indicating that spatial stabilizer redundancy effectively substitutes for further temporal repetition. Conversely, the `Bad' code ($d_{\mathrm{S}}=2$) shows no such convergence, with $P_{\mathrm{L}}$ fluctuating due to the parity constraints of majority voting. The `Good' code's higher absolute $P_{\mathrm{L}}$ is attributable to its larger blocklength ($n=126$ vs.\ $n=36$), which increases the frame error probability per cycle despite its superior measurement robustness.

%%%%%%%%%%%%%%%%%%%%%%%%%%%%%%%%%%%%%%%%%%%%%%%%%%%%
\section{Conclusion}
\label{sec:conclusion}

% In this work, we have established an algebraic framework for understanding single-shot decoding in coprime Bivariate Bicycle codes. We proved that the polynomial $g(z)$ plays a dual role: it determines the logical dimension of the quantum code and simultaneously generates the classical syndrome codes used to correct measurement errors.

% This duality reveals a fundamental design tension in the coprime BB ansatz. By the Singleton bound (Theorem~\ref{thm:singleton-syndrome}), the syndrome distance is upper-bounded by the stabilizer redundancy, which is directly coupled to the quantum rate. Consequently, ``good'' quantum codes with high rates necessarily have sparse redundancy, limiting their single-shot measurement correction radius.

% Despite this structural bottleneck, we demonstrated that by carefully selecting $g(z)$ to maximize the BCH distance of the syndrome code, one can construct short-blocklength BB codes ($N=63$) that support effective single-shot decoding. Our numerical results with BP+OSD confirm that these codes achieve logical error rates comparable to multi-round repetition schemes while significantly reducing the decoding cycle time.

% Future work should extend this analysis to circuit-level noise models, where hook errors may further constrain the effective distance, and explore generalized 2BGA constructions where the coupling between rate and redundancy can be relaxed.

%Concise:
This work establishes an algebraic framework for single-shot decoding in coprime Bivariate Bicycle codes, proving that the polynomial $g(z)$ simultaneously sets the quantum logical dimension and generates the classical syndrome codes for measurement correction. This reveals a fundamental trade-off: high-rate quantum codes inherently have limited stabilizer redundancy (per the Singleton bound), constraining their single-shot capability. However, we demonstrate that maximizing the BCH distance of $g(z)$ allows short codes ($N=63$) to achieve single-shot fidelity comparable to multi-round schemes, with significantly reduced latency. Future work should address circuit-level noise and explore generalized constructions to relax the rate-redundancy coupling.

\bibliographystyle{IEEEtran}
\bibliography{refs}

\end{document}